\documentclass{article}
 
\usepackage{amssymb, amsmath, amsthm, algorithm}
\usepackage{hyperref}
\usepackage{xspace} 
\usepackage{authblk}
\def\IF{{\bf if}\ }

\def\ELSE{{\bf else}\ }
\def\WHILE{{\bf while}\ }

\def\FOR{{\bf for}\ }

\def\RETURN{{\bf return}\ }
\def\blo{\noindent
\begin{tabular}{@{\quad}l@{\quad}}
\begin{minipage}{1in}
\begin{tabbing}
\qquad\=\qquad\=\qquad\=\qquad\=\qquad\=\qquad\=\qquad\=\kill}
\def\elo{\end{tabbing}\end{minipage}\\\end{tabular}}
\newenvironment{pseudocode}{\blo}{\elo}

\theoremstyle{definition}
\newtheorem{theorem}{Theorem}
\newtheorem{definition}[theorem]{Definition}
\newtheorem{lemma}[theorem]{Lemma}

\newcommand{\bfa}{\emph{bfa}\xspace}
\renewcommand{\emptyset}{\varnothing}
\newcommand{\poly}{\operatorname{poly}}

\newcommand{\capacity}{\operatorname{capacity}}
\newcommand{\cost}{\operatorname{cost}}
\newcommand{\size}{\operatorname{size}}
\newcommand{\var}[1]{\mathit{#1}} 
\newcommand{\func}[1]{\textsf{#1}} 
\newcommand{\x}[3][x]{#1(#2,#3)}

\title{The Subset Assignment Problem for Data Placement in Caches}

\author[1]{Shahram Ghandeharizadeh\thanks{shahram@dblab.usc.edu}}
\author[2]{Sandy Irani\thanks{irani@ics.uci.edu}}
\author[3]{Jenny Lam\thanks{jenny.lam01@sjsu.edu}}
\affil[1]{Department of Computer Science, University of Southern California}
\affil[2]{Department of Computer Science, University of California, Irvine}
\affil[3]{Department of Computer Science, San Jos\'{e} State University}

\begin{document}

\maketitle

\begin{abstract}
We introduce the subset assignment problem in which items of varying sizes are placed in a set of bins with limited capacity. Items can be replicated and placed in any subset of the bins. Each (item, subset) pair has an associated cost. Not assigning an item to any of the bins is not free in general and can potentially be the most expensive option. The goal is to minimize the total cost of assigning items to subsets without exceeding the bin capacities. This problem is motivated by the design of caching systems composed of banks of memory with varying cost/performance specifications. The ability to replicate a data item in more than one memory bank can benefit the overall performance of the system with a faster recovery time in the event of a memory failure. For this setting, the number $n$ of data objects (items) is very large and the number $d$ of memory banks (bins) is a small constant (on the order of $3$ or $4$). Therefore, the goal is to determine an optimal assignment in time that minimizes dependence on $n$. The integral version of this problem is NP-hard since it is a generalization of the  knapsack problem. We focus on an efficient solution to the LP relaxation as the number of fractionally assigned items will be at most $d$. If the data objects are small with respect to the size of the memory banks, the effect of excluding  the fractionally assigned data items from the cache will be small. We give an algorithm that solves the LP relaxation and runs in time $O(\binom{3^d}{d+1} \poly(d) n \log(n) \log(nC) \log(Z))$, where $Z$ is the maximum item size and $C$ the maximum storage cost.
\end{abstract}

\section{Introduction}

We define a combinatorial optimization problem which we call the \emph{subset assignment problem}. An instance of this problem consists of $n$ items of varying sizes and $d$ bins of varying capacities. Any  item can be replicated and assigned to multiple bins. 
A problem instance also includes $n\cdot 2^d$ cost parameters which denote for each item and each subset of the bins, the cost of storing copies of the item on that subset of the bins. 
The objective  is to find an assignment of items to subsets of bins which minimizes the total cost  
subject to the constraint that the sum of the sizes of items assigned to each bin does not exceed the capacity
of the bin.

The costs do not necessarily exhibit any special properties, although we do assume that they are non-negative. For example, we do not assume that cost necessarily increases or decreases
the more bins an item is assigned to. Assigning an item to the empty set, which corresponds to not assigning the item to any of the
bins, is not free in general and can potentially be the most expensive option for an item.

The subset assignment problem is a natural generalization of the multiple knapsack problem (MKP), in which each item can only be stored on a single bin. 
The book by Martello and Toth~\cite{MarTot-KnapProb-1990} and the more recent book by Kellerer et al.~\cite{KelPfePis-KnapProb-2004} both devote a chapter to MKP. 
A restricted version of the subset assignment problem in which each item can only be stored in a single bin
corresponds to the multiple knapsack problem (MKP). MKP
 is known to be NP-complete but does have a polynomial time approximation scheme~\cite{CheKha-SODA-2000}.
For the application we are interested in, the number of items $n$ is very large (on the order of billions) and the number of bins is a small constant (on the order of $3$ or $4$). Furthermore, the size of even the largest item is small with respect to the capacity of the bins.
Since there is an optimal solution for the linear programming relaxation in which at most $d$ items are fractionally placed, the effect of excluding the fractionally assigned items from the cache is negligible.  Therefore, we focus on an efficient solution to the linear programming relaxation. 

The linear relaxation of MKP can be expressed as a minimum cost flow on a bipartite graph, a classic and well studied problem in the literature~\cite{AhuMagOrl-NetwFlow-1993}.  Tighter analysis for the case of minimum cost flow on an imbalanced bipartite graph ($n >> d$) is given in~\cite{GusMarFer-SIAMJComp-1987} and improved in~\cite{AhuOrlSteTar-SIAMJComp-1994}.
Naturally, the goal with highly imbalanced bipartite graphs (which
corresponds to the situation
in the subset assignment problem in which  the number of items is much larger than the number of bins) is to minimize dependence on $n$, even at the expense of greater dependence on $d$.

We propose an algorithm for the linear programming relaxation of the subset assignment problem that is similar in structure to cycle canceling algorithms for min-cost flow and is also inspired by the concept of a bipush, which is central to the tighter analysis~\cite{GusMarFer-SIAMJComp-1987} and~\cite{AhuOrlSteTar-SIAMJComp-1994}. The analysis shows that our  algorithm runs in $O(f(d) \poly(d) n \log(n) \log (nC) \log (Z))$, where $C$ is the maximum cost of storing an item on any subset of the bins and $Z$ is the maximum size of any item.
The function $f(d)$ is defined to be the number of distinct sets of vectors $\{\vec{v}_1, \ldots, \vec{v}_r \}$
where $\vec{v}_i \in \{-1, 0, 1\}^d$ and the  solution $\vec{\alpha}$ to the equation
$\sum_{i=1}^r \alpha_i  \vec{v}_i = \vec{0}$ with  $\alpha_1 = 1$ is unique and positive ($\vec{\alpha} > 0$).
In order for the solution $\vec{\alpha}$ to be unique, the first $r-1$ vectors must be linearly independent and therefore $r \le d+1$.
If $r < d+1$, the set can be uniquely expanded to a set of size $d+1$
such that the solution to $\sum_{i=1}^r \alpha_i  \vec{v}_i = \vec{0}$ with  $\alpha_1 = 1$ is unique and non-negative ($\vec{\alpha} \ge 0$).
This observation gives an  upper bound of $\binom{3^d}{d+1}$ for $f(d)$,
 hence the running time
of  $O(\binom{3^d}{d+1} \poly(d) n \log(n) \log (nC) \log (Z))$.
Numerical simulation has shown that $f(3) = 778$ and $f(4) = 531,319$.
Since the problem specification requires $n\cdot 2^d$ cost parameters, an exponential dependence on $d$ is unavoidable.
A direction for future research is to reduce the dependence on $d$ from exponential in $d^2$ to exponential in $d$.

A reasonable assumption for the cache assignment problem (the motivating application for the subset
assignment problem) is that it is never advantageous to replicate an item in more than two bins.
Under this assumption, the vectors can have at most 4 non-zero entries. So the function $f(d)$ is bounded by $d^{O(d)} \poly(d)$ resulting in an
overall running time of
$O(d^{O(d)} \poly(d) n \log(n) \log (nC) \log (Z))$.

In a linear programming formulation of the problem there are $2^d n$ variables and $n + d$ constraints. The best polynomial
time algorithms to solve a general instance of linear programming require time at least cubic in $n$ which for the
values we consider is prohibitively large. (For example, Karmarkar's algorithm requires $O(N^{3.5}L)$ operations
in which $N$ is the number of variables and all input numbers can be encoded with $O(L)$ digits \cite{Kar84}.)
Our algorithm is closer to the Simplex algorithm in that after each iteration, the current solution is a basic feasible solution.
However, the algorithm does not necessarily traverse the edges of the simplex. The algorithm selects an optimal local
improvement which, in general, can result in a solution which is not a basic feasible solution and then restorse the solution to
a basic feasible solution without increasing the cost. It is not clear how to bound the running time of any implementation of
the Simplex algorithm in which the algorithm is bound to traverse edges of the simplex.
Since the problem can be formulated as a packing problem, there is a randomized approximation algorithm whose solution
is within a factor of $1 +\epsilon$ of optimal and whose running time is 
$O(2^d \poly(d) n \log n / \epsilon^2)$ \cite{KY14}.



\subsection{Motivation}

The subset assignment problem is motivated by the problem of managing a multi-level cache. 
Although caches are used in many different contexts, we are particularly interested in the use of caches to augment database management systems. Query results (called {\em key-value pairs}) are stored in the cache so that the next time the query is issued, the result can be retrieved from memory instead of recomputed from scratch. Typically key-value pairs vary in size as they contain different types of data.
Furthermore, re-computation can vary dramatically, depending on the application. 
In some applications, a key-value pair is the result of hours of data-intensive computation.
If the key-value pair does not reside in the cache (corresponding to allocating the item to the empty set in our
formulation), this computation cost would be paid every time the key-value pair is accessed.
Memcached, currently the most popular key-value store manager, is used by companies such as Facebook~\cite{saab08}, Twitter and Wikipedia. Today's memcached uses DRAM for fast storage and retrieval of key-value pairs. However, using a cache that consists of a collection of memory banks with different characteristics can potentially improve cost or performance~\cite{flexMemTech}.

We model a  sequence of requests to data items (key-value pairs) in the cache
as a stream of independent events as do social networking benchmarks such as BG~\cite{sumita13} and LinkBench~\cite{linkbench}. 
Cache management (or paging) under i.i.d. request sequences is also well studied in the theory literature \cite{Sta-PerfEval-2001,JelRad-INFOCOM-2003}.
If query and update (read and write) statistics are known in advance, the optimal policy is a static placement of data items in the memory banks that minimizes the expected time to service each request. A static placement can have much better performance over adaptive online algorithms if the request frequencies are stable \cite{camp14, campTechreport}. Since the popularity of queries do  vary over time, a static placement would need to be recomputed periodically based on recent statistics followed with a reorganization of key-value pairs across memory banks.

With the advent of Non-Volatile Memory (NVM) such as PCM, STT-RAM,  NAND Flash, and the (soon to be released) Intel X-Point,
cache designers are provided with a wider selection of memory types with different performance, cost and reliability characteristics. 
The relative read/write latency and bandwidth for different memory types vary considerably.
An important challenge in computer system design is how to effectively design caching middleware that leverages these new choices~\cite{kimpcm2014, flexMemTech, Nan15}. The survey in \cite{Nan15} makes the case that the advent of new storage
technologies significantly changes the standard assumptions in system design and leveraging such technologies will require more sophisticated
workload-aware storage tiering.

In this paper, the cache is composed of a small number of memory banks each of which is a different  type of memory.
The goal is to find an optimal placement of data items in the cache.
Our model also takes into account that a memory bank can fail due to a power outage or hardware failure. (Non-volatile memories do not lose their content during  a power outage but they can experience hardware failures). If a memory bank fails, its contents must be restored, either all at once or over time. In this case, it may be advantageous to store a data item on more than one memory bank so that the data can be more easily recovered in the event that one or more memory banks fail. On the other hand, maintaining multiple copies of a key-value pair can be costly if they must be frequently updated. We express these different trade-offs in an optimization problem by allowing  a key-value pair to be replicated and stored on any subset of the memory banks. The $\emptyset$ option represents not keeping  the key-value pair in the cache at all and recomputing the result from the database at every query, an option that can  be computationally very costly.
Simulation results from \cite{flexMemTech} show that it can be advantageous to store copies of data items in more than one
memory bank to speed up recovery time, although it depends on the read and write frequencies of the data items
as well as the failure rates of the memory. 

\cite{flexMemTech} gives a detailed description for how the memory parameters and request frequencies translate into costs and uses the model to study a closely related problem in which one is given a fixed budget as well as the price for the different types of memory. The goal is to determine the optimal amount of each type of memory to purchase as well as the optimal placement of key-value pairs to memory banks that minimizes expected service time subject to the overall budget constraint. The algorithm in \cite{flexMemTech} is implemented and evaluated
using traces generated by a standard social networking benchmark~\cite{sumita13}.
In this paper, we consider the situation in which the design of  the cache is already determined in that there is a set of memory banks whose capacities are given as part of the problem input. The goal is to place each item on a subset of the memory banks so that the capacities of each memory bank is not exceeded and the total cost is minimized.

In both cases, the objective function is expected  service time for all the items. The cost of serving an item $p$ located on subset $S$ of the memory banks is  a sum of three terms: the total expected time to serve read requests to $p$, the total expected time to 
serve write requests to $p$, and the total expected time needed to restore $p$ in the event that one or more memory bank in $S$ fails. More explicitly:
\begin{align*}
\cost(p, S) &= \text{read-freq}(p) \cdot \text{read-time}(p, S)\\
& + \text{write-freq}(p) \cdot \text{write-time}(p,S)\\
&+ \sum_{F} \text{fail-freq}(F) \cdot \big(\text{read-time}(p, S \setminus F) + \text{write-time}(p, F \cap S).\big)
\end{align*}
The functions read-freq($p$) and write-freq($p$) represent the probability of a read or write request to $p$. The function fail-freq($F$) is the probability that all the memory banks in subset $F$ fail. 
On a request to read an item $p$, item $p$ can be obtained from any of the copies of $p$ in the cache.
Therefore,
read-time($p,S$) represents the time to read $p$ from the memory bank in $S$ that provides the fastest read time.
The time to read a data item from a memory bank depends on the size of the item as well as the latency and bandwidth for reading
from that type of memory.
Updating an item, on the other hand, requires updating every copy of that item in the cache.
Therefore,
write-time($p,S$) is the maximum time to write $p$ to any of the memory banks in $S$, assuming that
 writing $p$ to its multiple destinations is  done in parallel. (If writing is done sequentially, then write-time($p,S$) is
the sum of the write times over  all the memory banks in $S$).
The time to write  a data item to  a memory bank depends on the size of the item as well as the latency and bandwidth for writing
to  that type of memory.
Recovering from failure could involve reading from those memory banks that still have a copy of $p$ and rewriting them to those that lost it. 


Since the relative read/write frequencies for items and read/write times for memories vary significantly,
there is no useful structure to exploit in modeling the cost of assigning an item to a subset
which is why they are assumed to be arbitrary values given as part of the input to the subset assignment problem.
However, based on the empirical failure rates of the memory technologies, it is reasonable to assume that two memory banks will
never  fail at the same time. Under this assumption, there is no need to keep more than two copies of an item in the cache
and we can restrict the data placements for an item to subsets  of size one or two.
The problem is addressed in this paper in its full generality although a better bound on the running time can be obtained
with this restriction.

\section{Problem Definition}

There are $n$  items and each item $p$ has a given  $\size(p)$.
There are  $d$ bins $\mathcal{B} = \{b_1, \ldots , b_d\}$. 
Each bin $b$ has a given  $\capacity(b)$.
An item can be replicated and placed on any subset of the memory banks
$S \subseteq \mathcal{B}$. 
We call $S$ a {\em placement option} for an item.
Placing $p$ on  $S$ has cost denoted by $\cost(p, S) \geq 0$.
A placement of items to memory banks
is described by a set of $n \cdot 2^d$ variables $\x{p}{S} \geq 0$ with the constraint that for each $p$,
\begin{equation}
\label{eq:condition1}
\sum_S \x{p}{S} = \size(p).
\end{equation}
Also the capacity of each bin cannot be exceeded, so for each $b$,
\begin{equation}
\label{eq:condition2}
\sum_{S \ni b} \sum_p \x{p}{S} \leq \capacity(b).
\end{equation}
The goal is to minimize
\[
\sum_{p} \sum_{S} \cost(p, S) \x{p}{S},
\]
subject to the condition that all $x(p,S) \geq0$, (\ref{eq:condition1}) and (\ref{eq:condition2}) above.

The placement option $\emptyset$, corresponding to not placing an item in
any of the bins, is  an option for every $p$, so the problem always has a feasible solution.
For each bin $b$, we will add an extra item $p$ whose size is $\capacity(b)$.
For each added $p$, $\cost(p, \emptyset) = \cost(p, \{b\}) = 0$. For all other
$S \subseteq \mathcal{B}$,
$\cost(p, S) = \infty$. We assume that the pages are numbered so that the extra item
for bin $b_i$ is $p_i$. With the additional items, we can assume that every solution
under consideration has every bin filled exactly to capacity since any extra space in $b_i$
can be filled with $p_i$ without changing the cost of the solution. 
Therefore we require that for each $b$, $\sum_{S \ni b} \sum_p \x{p}{S} = \capacity(b)$.
An assignment which satisfies the equality constraints on the bins is called
\emph{perfectly filled}. 


\section{Preliminaries}

Our algorithm
starts with a feasible, perfectly filled solution and improves  the assignment in a series
of small steps, called augmentations.
The augmentations, a generalization of a negative cycle in min-cost flow, always maintain the condition that the current
assignment is feasible and perfectly
filled. In each iteration
the algorithm finds an augmentation that approximates the best possible augmentation
in terms of the overall improvement in cost.
An augmentation is a linear combination of moves in which mass is moved
from $x(p,S)$ to $x(p,T)$ for some item $p$.
Each move gives rise to a $d$-dimensional vector over $\{-1, 0, 1\}$ that 
denotes the net increase or decrease to each bin as a result of  the move.
We require that the linear combination of vectors for an augmentation
equal $\vec{0}$ in order to maintain the condition that the bins are perfectly filled.
The profile for an augmentation is the set of vectors corresponding
to the moves in that augmentation.
In order to find  a good augmentation, we exhaustively search over all
profiles 
and then find a good set of actual moves that correspond to each profile.
Exhaustively searching over all profiles introduces a factor of $f(d)$, the
number of distinct profiles which is at most $\binom{3^d}{d+1}$.

In order to bound the number of iterations, we also need to establish that 
there is an augmentation that improves the cost by a significant factor.
For flows, this is accomplished by showing that the difference between the current 
solution and the optimal solution can be decomposed into at most $m$ simple 
cycles, where $m$ is the number of edges in the network. If $\Delta$ is the difference between the current and optimal cost, then
there is a cycle that improves the cost by at least $\Delta/m$.
We proceed in a similar way, showing that the difference between two assignments
can be decomposed into at most $2(n+d)$ augmentations any of which can be applied
to the current assignment. Therefore there is an augmentation that improves
the cost by at least $\Delta/2(n+d)$.

\subsection{Augmentations}

For $S \subseteq \mathcal{B}$, $\vec{S}$ is a $d$-dimensional vector whose $i^{th}$ coordinate
is $1$ if $b_i \in S$ and is $0$ otherwise. 
Let $\mathcal{V}$ be the set of all length $d$ vectors over $\{-1, 0, 1\}$.
A set $V \subseteq \mathcal{V}$  is said to be \emph{minimally dependent}
if  $V$ is linearly dependent and no proper subset of $V$ is linearly dependent.
If $V = \{\vec{v}_1, \ldots, \vec{v}_r\}$ is minimally dependent, then the values $\alpha_1, \ldots, \alpha_r$ such that
$\sum_{i=1}^r \alpha_i \vec{v}_i = \vec{0}$
are unique up to a global constant factor. In order to make a unique vector $\vec{\alpha}$, we always
maintain the convention that $\alpha_1 = 1$.
A minimally dependent set $V$ is said to be \emph{positive} if the associated vector $\vec{\alpha} > \vec{0}$.

A \emph{move} is defined by a triplet $(p, S, T)$ that represents the possibility of moving mass from
$\x{p}{S}$ to $\x{p}{T}$.
The \emph{profile} for a set of moves 
\[
\{(p_1,S_1, T_1), \ldots, (p_r, S_r, T_r)\}
\]
is the set of vectors $\{ (\vec{T}_1 - \vec{S}_1), \ldots, (\vec{T}_r - \vec{S}_r) \}$.
Note that the vector $\vec{T} - \vec{S}$ represents the net increase or decrease to each bin
that results from
moving one unit of mass from $\x{p}{S}$ to $\x{p}{T}$ for some $p$.
A set of moves is called an \emph{augmentation} if the set of vectors in its profile is
minimally dependent and positive. Note that an augmentation contains at most $d+1$ moves.

An augmentation $\mathcal{A} = \{(p_1,S_1, T_1), \ldots, (p_r, S_r, T_r)\}$ 
can be applied to a particular assignment $\vec{x}$ if for every 
$i = 1, \ldots r$,  $\x{p_i}{S_i} > 0$.
Let $\vec{\alpha}$ be the unique vector of values such that $\alpha_1 = 1$
and $\sum_{j=1}^r \alpha_j (\vec{T}_j - \vec{S}_j) = \vec{0}$.
If the augmentation is applied with magnitude $a$ to $\vec{x}$, then for every
$(p_j, S_j, T_j)  \in \mathcal{A}$,
$\x{p_j}{S_j}$ is replaced with $\x{p_j}{S_j} - a \alpha_j$
and $\x{p_j}{T_j}$ is replaced with $\x{p_j}{T_j} + a \alpha_j$.
The cost vector for an augmentation is $\vec{c}$, where
$c_j = \cost(p_j, T_j) - \cost(p_j, S_j)$.
The cost associated with applying the augmentation with magnitude $a$
is $a (\vec{c} \cdot \vec{\alpha})$.
Since the goal is to minimize the cost,  we only apply augmentations whose cost is
negative.

For augmentation $\mathcal{A} = \{(p_1,S_1, T_1), \ldots, (p_r, S_r, T_r)\}$,
let $\mathcal{S}(\mathcal{A})$ be the set of all pairs $(p, S)$ such that for some $i$,
$p = p_i$ and $S = S_i$.
For each $(p, S) \in \mathcal{S}(\mathcal{A})$, define 
\[
\alpha(p, S) = \sum_{i: p_i = p, S_i = S} \alpha_i.
\]
The maximum magnitude with which the augmentation
$\mathcal{A}$ can be applied to $\vec{x}$ is 
\[
\min_{(p,S) \in \mathcal{S}(\mathcal{A})} \frac{\x{p}{S}}{\alpha(p, S)}.
\]
The following lemma is analogous to the fact  for flows that says
 there is always a cycle in the network representing the difference
between two feasible flows. The proof is given in the Appendix.

\begin{lemma}
\label{lem:vecDiff}
Let $\vec{x}$ and $\vec{y}$ be two feasible, perfectly filled assignments to the same instance of the subset assignment problem.
Then there is an augmentation that can be applied to $\vec{x}$ that consists only of moves of the form
$(p, S, T)$ where $\x{p}{S} > \x[y]{p}{S}$ and $\x{p}{T} < \x[y]{p}{T}$.
\end{lemma}

\subsection{Basic Feasible Assignments}

An item is said to be \emph{fractionally assigned} if there are two subsets
$S \neq S'$, such that $x(p,S) > 0$ and $x(p,S') > 0$.
If items can only be assigned to single bins as in the standard assignment problem,
then it follows from total unimodularity that the optimal solution is integral,
assuming that all the input values are integers. 
For  the subset assignment problem, the optimal solution may not be integral,
even if all the input values are integers. Here is an example in which
the item sizes and bin capacities are $1$, but an optimal
solution must have fractionally assigned items: we have two items $p$ and $q$ and two bins $b$ and $c$ with costs
\begin{align*}
  &\cost(p,\emptyset) = 1,\ \ \cost(p, \{b, c\}) = 0, &&\cost(q, \emptyset) = \cost(q, \{b, c\}) = C,\\
  &\cost(p, \{b\}) = \cost(p, \{c\}) = C, &&\cost(q, \{b\}) = \cost(q, \{c\}) = 0,
\end{align*}
where $C$ is a large number. The optimal assignment is to equally distribute $p$ over $\{b, c\}$ and $\emptyset$, and to equally distribute $q$ over $\{b\}$ and $\{c\}$.

The linear programming formulation of the subset assignment problem has $n+d$
constraints. $n$ constraints enforce that each $p$ must be assigned:
\[
\sum_S x(p, S) = size(p).
\]
The other $d$ constraints, say that each bin
must be exactly filled to capacity. 
Therefore, any basic feasible solution to the linear programming formulation
of the subset assignment problem has at most $n+d$ non-zero
variables.
Since for every $p$, there is at least one $S$ such that $x(p,S) > 0$
and $n >> d$,
we know at least $n - d$
of the items will not be fractionally assigned
because they have only one $S$ such that $x(p,S) > 0$.
The number of variables $x(p,S)$ such that $0 < x(p,S) < \size(p)$
is at most $2d$, so the number of fractionally assigned items is
at most $d$. 

The criteria for a feasible solution to be a basic feasible solution
is that once the variables are chosen that will be positive, there is exactly
one way to assign values to those variables so that all the constraints are satisfied.
Suppose we have a feasible  assignment $\vec{x}$. First place the items in bins
that are not fractionally assigned. If $\vec{x}$ is a basic feasible solution,
then there is a unique way to place the remaining items so that
the bins are filled exactly to capacity. 
We rephrase the definition of a basic
feasible solution in the language of the subset assignment problem
and prove the same facts about the new definition.

Consider a feasible assignment $\vec{x}$.
Let $P_{\text{frac}}$ be the set of data items that are fractionally assigned.
Let $X_{\text{frac}}$ be the set of variables $\x{p}{S}$ such that $0 < \x{p}{S} < \size(p)$.
Let $P_{\text{int}}$ be the set of items that are assigned to exactly one subset.
That is $p \in P_{\text{int}}$ if $\x{p}{S} \in \{0, \size(p)\}$ for all $S$.

\begin{definition}
For each $p \in P_{\text{frac}}$ select one $S$ such that $\x{p}{S} > 0$. Denote the selected set
for $p$ by $S_p$. Let $X$ be the set of variables $\x{p}{S}$ such that $S \neq S_p$ and
$0 < \x{p}{S} < \size(p)$. 
Let $V$ be the set of vectors $\vec{S} - \vec{S}_p$ for each $\x{p}{S} \in X$.
Then $\vec{x}$ is a \emph{basic feasible assignment} (\bfa) if and only if $V$ is linearly independent.
\end{definition}

Although the definition for a basic feasible assignment was given in terms of
a particular choice of $S_p$'s, the property of being a \bfa does not depend on this choice.

\begin{lemma}
The condition of being a \bfa does not depend on the choice of $S_p$, for 
$p \in P_{\text{frac}}$.
\end{lemma}

\begin{proof}
Let $p \in P_{\text{frac}}$ and let $\{S_1, \ldots, S_r\}$ be the subsets such that
$x(p,S_j) > 0$.
Suppose that $S_p$ is chosen to be $S_i$. 
Select any two $S_j \neq S_k$. Since $(\vec{S}_j - \vec{S}_k) = (\vec{S}_j - \vec{S}_p) - (\vec{S}_k - \vec{S}_p),$
the space spanned by all $(\vec{S}_j - \vec{S}_k)$ for $S_j \neq S_k$ is equal
to the space spanned by all $(\vec{S}_j - \vec{S}_p)$ for $S_j \neq S_p$.
The space spanned by all $(\vec{S}_j - \vec{S}_k) \neq \vec{0}$
is independent of the choice of $S_p$.
\end{proof}

\begin{lemma}
\label{lem:bfaFracBound}
If $\vec{x}$ is a \bfa, then the number of variables in $X_{\text{frac}}$ is at most $2d$ and the number of fractionally assigned items is at most $d$.
\end{lemma}

\begin{proof}
Since $|X| = |V|$, and $V$ must be linearly independent for any \bfa,
it must be that if $\vec{x}$ is a \bfa, then $|X| \leq d$. 
The set of fractionally assigned variables
($X_{\text{frac}}$)
includes all the $\x{p}{S_p}$ for $p \in P_{\text{frac}}$ and $X$.
For each $\x{p}{S_p}$, there is at least one variable in $X$.
Therefore the 
number of variables such that $0 < \x{p}{S} < \size(p)$ in any \bfa
is at most $2d$. 
\end{proof}

The process
\func{Restore}, given in the Appendix, takes 
an assignment $\vec{x}$ which may not be
a \bfa and restores  it  to an assignment which is a \bfa.
The process maintains the condition that the current assignment is
feasible and perfectly filled.
If the set $V$ is linearly dependent, a linear combination of the 
moves $(p, S_p, S)$ is chosen for each $\x{p}{S} \in X$
such that applying the linear combination of moves keeps the bins
perfectly filled. 
Since $p$ has some weight on $S_p$ and some weight on $S$,
 all the moves can be applied in either the forward or reverse direction.
(A negative coefficient denotes applying a move in the reverse direction.)
We choose a direction for the linear combination of moves such that the cost does not increase.
The combination of moves is applied 
until either $\x{p}{S}$ or $\x{p}{S_p}$ becomes $0$ for one
of the moves represented in $V$.
Thus, the cost of the assignment does not increase and the number of fractionally assigned
variables decreases by at least one.
The process continues until $V$ is linearly independent.

\begin{lemma}
There is an optimal solution that is also a \bfa.
\end{lemma}

\begin{proof}
Start with an optimal assignment  $\vec{x}$ which may not be a \bfa.
Apply \func{Restore} to $\vec{x}$. The resulting assignment is a \bfa.
And since the cost of $\vec{x}$ does not increase, $\vec{x}$ is
still optimal.
\end{proof}

\section{The Algorithm}

The algorithm we present proceeds in a series of iterations. In each iteration, we apply an augmentation
 to the current assignment. 
Since the resulting assignment may no longer be a \bfa,  we then apply 
\func{Restore} to turn the solution back into a \bfa.

\begin{algorithm}
\begin{pseudocode}
$\x{p}{S} = 0$, for all $p$ and $S$.\\
$\x{p_i}{\{b_i\}}\!=\! \capacity(b_i)$, for $i\!=\! 1, ..., d$.\! (Fill\! each\! bin\! with\! the\! ``extra''\! items.)\\
$\x{p_j}{\emptyset} = \size(p_j)$, for $j > d$. (All\! the\! ``original''\! items\! start\! outside\! the\! bins.)\\
$\mathcal{P} = \func{Preprocess}(d)$\\
$\mathcal{A} = \func{FindAugmentation}(\vec{x})$\\
\WHILE $\mathcal{A} \neq \emptyset$\\
\> Apply $\mathcal{A}$ to $\vec{x}$ with the largest possible magnitude\\
\> $\func{Restore}(\vec{x})$. (Transform $\vec{x}$ into a \bfa.)\\
\> $\mathcal{A} = \func{FindAugmentation}(\vec{x})$
\end{pseudocode}
\caption{\func{MainLoop}}
\label{alg:main}
\end{algorithm}

Note that it is possible to find an augmentation that moves from a \bfa to another \bfa directly.
This is essentially what the simplex algorithm does.
However, not every augmentation results in a \bfa.
The augmentations that do result in a \bfa
must include the moves that shift mass  between the fractionally assigned items.
(These moves correspond to the vectors $V$ described in the definition of a \bfa). 
Restricting the augmentation in this way may result in a sub-optimal augmentation.
For example, those augmentations could require decreasing a variable that is 
already very small in which case the augmentation can not be applied with very large magnitude.
So we allow the algorithm to select from the set of all augmentations
to get as much benefit as possible, and then
move the assignment to a \bfa. 


\subsection{Finding an augmentation that is close to the best possible}

The first step is a preprocessing step in which every possible augmentation profile
is generated. This consists of generating every minimally dependent subset $V$ of
$\mathcal{V}$ and its associated $\vec{\alpha}$.
$\func{Preprocess}$ (shown in the Appendix)
runs in time $O(\binom{3^d}{d+1}\poly(d))$. The number of distinct augmentation profiles returned by the procedure is at most $\binom{3^d}{d+1}$.

Given an augmentation profile $V = \{ \vec{v}_1, \ldots, \vec{v}_r\}$,
the goal is to find an augmentation whose profile matches $V$
and can be applied with a magnitude that gives close to the best possible improvement.
For each vector $\vec{v} \in \mathcal{V}$, we maintain a data structure with every move
$(p, S, T)$ such that $\vec{T} - \vec{S} = \vec{v}$ and
$\x{p}{S} > 0$. We will call the set of all such moves
$\var{Moves}(\vec{v})$.
The data structure should be able to answer queries of the form: given $x_0$,
find the move $(p, S, T)$ such that $\cost(p,T) - \cost(p,S)$ is minimized
subject to the condition that $\x{p}{S} \geq x_0$.
These kind of queries can be handled by an augmented binary search tree in
logarithmic time~\cite{CLRS-IntrAlgo-2009}.

For a given \bfa $\vec{x}$ and augmentation $\mathcal{A}$, one
can calculate the maximum possible magnitude $a$ with which
$\mathcal{A}$ can be applied to $\vec{x}$. We will make use
of upper and lower bounds for the value $a$
for any augmentation and \bfa combination.
Call these values $a_{\max}$ and $a_{\min}$.
Round $a_{\min}$ down so that $a_{\max}/a_{\min}$ is a power of $2$.
The while loop in procedure \func{FindAugmentation} runs for
$\log(a_{\max}/a_{\min})$ iterations. 

\begin{algorithm}
\begin{pseudocode}
$\var{BestCost} = 0$\\
$\mathcal{A} = \emptyset$\\
\FOR each augmentation profile $V = \{ \vec{v}_1, \ldots, \vec{v}_r \}$ and vector $\vec{\alpha}$\\
\> $a = a_{\max}/2$\\
\> \WHILE $a \geq a_{\min}$\\
\>\> \FOR $i = 1, \ldots, r$\\
\>\>\> Let $(p_i, S_i, T_i)$ be the move with the smallest cost among \\
\>\>\>\> moves in $\var{Moves}(\vec{v}_i)$ such that $\x{p_i}{S_i} \geq a \cdot \alpha_i$. \\
\>\>\> $c_i = \cost(p_i, T_i) - \cost(p_i, S_i)$\\
\>\> $\var{CurrentCost} = \sum_{i=1}^r a \cdot c_i \cdot \alpha_i$\\
\>\> \IF $\var{CurrentCost} < \var{BestCost}$\\
\>\>\> $\var{BestCost} = \var{CurrentCost}$\\
\>\>\> $\mathcal{A} = \{ (p_1, S_1, T_1), \ldots, (p_r, S_r, T_r)\}$\\
\>\> $a = a/2$\\
\RETURN $\mathcal{A}$
\end{pseudocode}
\caption{\func{FindAugmentation}($\vec{x}$)}
\label{alg:find-aug}
\end{algorithm}

For an augmentation $\mathcal{A}$ that can be applied to assignment $\vec{x}$ with
magnitude $a$, the total change in cost is
denoted by $\cost(\mathcal{A}, \vec{x}, a)$.
Recall that since we are minimizing cost we will only apply an augmentation
if the total change in cost is less than $0$.

\begin{lemma}
\label{lem:findAug}
  Let $\mathcal{A}_1$ be the augmentation returned by $\func{FindAugmentation}(\vec{x})$ and $\mathcal{A}_2$ be another augmentation. If $a_1$ and $a_2$ are the maximum magnitudes with which $\mathcal{A}_1$ and $\mathcal{A}_2$ can be applied to $\vec{x}$, then $2d \cdot \cost(\mathcal{A}_1, \vec{x}, a_1 ) \leq\cost(\mathcal{A}_2 , \vec{x}, a_2 )$.
\end{lemma}

\begin{proof}
Let $V_2$ be the profile for $\mathcal{A}_2$. Let $\vec{\alpha}$ be the vector associated with
the profile $V_2$. 
Let ${\bar a}$ be the value of the
form $a_{\max}/2^j$ such that $2 {\bar a} > a_2 \geq {\bar a}$.
There is an iteration inside the while loop of $\func{FindAugmentation}(\vec{x})$
in which the augmentation profile is $V_2$ and the value for $a$ is ${\bar a}$.
The augmentation constructed in this iteration will be called $V_3$.
The moves in $\mathcal{A}_2$ and in $\mathcal{A}_3$ are $\{(p_1^{(2)}\!, S_1^{(2)}, T_1^{(2)}),... , (p_r^{(2)}\!, S_r^{(2)}, T_r^{(2)})\}$ and $\{(p_1^{(3)}\!, S_1^{(3)}, T_1^{(3)}),..., (p_r^{(3)}\!, S_r^{(3)}, T_r^{(3)})\}$, respectively.

Note that since $V_2$ can be applied to $\vec{x}$ with magnitude $a_2$, it must be the case that
for $i = 1, \ldots, r$,
$\x{p_i^{(2)}}{S_i^{(2)}} \geq \alpha_i a_2$ because applying the moves involves removing
$\alpha_i a_2$ from $\x{p_i^{(2)}}{S_i^{(2)}}$.
Since $a_2 \geq {\bar a}$,
$\x{p_i^{(2)}}{S_i^{(2)}} \geq \alpha_i {\bar a}$.
The move $(p_i^{(3)}, S_i^{(3)}, T_i^{(3)})$ is chosen to be the move with minimum cost
such that $\x{p_i^{(3)}}{S_i^{(3)}} \geq \alpha_i {\bar a}$.
Therefore
the cost of $(p_i^{(3)}, S_i^{(3)}, T_i^{(3)})$ is at most the cost of
$(p_i^{(2)}, S_i^{(2)}, T_i^{(2)})$.
The value of the variable $\var{CurrentCost}$ for that iteration is 
\begin{eqnarray*}
\var{CurrentCost}_3 
& = & \bar{a} \sum_{i = 1}^r \alpha_i \left[\cost(p_i^{(3)}, T_i^{(3)}) - \cost(p_i^{(2)}, S_i^{(3)})\right]\\
& \leq & \bar{a} \sum_{i = 1}^r \alpha_i \left[\cost(p_i^{(2)}, T_i^{(2)}) - \cost(p_i^{(2)}, S_i^{(2)})\right]\\
& \leq & \frac{a_2}{2} \sum_{i = 1}^r \alpha_i \left[\cost(p_i^{(2)}, T_i^{(2)} - \cost(p_i^{(2)}, S_i^{(2)}) \right]\\
& =& \frac{1}{2}\cost(\mathcal{A}_2 , \vec{x}, a_2 )
\end{eqnarray*}

Let $\var{CurrentCost}_1$ be the value of the variable
$\var{CurrentCost}$ and $a'$ the value of the variable $a$
during the iteration in which the augmentation $\mathcal{A}_1$ is considered. 
Since $\mathcal{A}_1$ was selected by
$\func{FindAugmentation}$, $\var{CurrentCost}_1 \leq \var{CurrentCost}_3$.
It remains to show that the maximum magnitude with which $\mathcal{A}_1$ can be applied
is at least $a'/d$ and therefore the actual change in cost at most
$\var{CurrentCost}_1/d$. 

Let $V_1$ be the profile for $\mathcal{A}_1$ and let $\vec{\beta}$ be the vector associated with
profile $V_1$.
Since we are now only referring to one augmentation, we omit the subscripts and call the moves
in $\mathcal{A} = \{(p_1, S_1, T_1), \ldots, (p_r, S_r, T_r)\}$.
We are guaranteed by the selection of the move $(p_i, S_i, T_i)$ that for every $i$,
$\x{p_i}{S_i}/\beta_i \geq a'$.
Let $\beta_{p, S}^{\text{sum}}$ and $\beta_{p, S}^{\max}$ denote the sum and maximum over all $\beta_i$ such that $p_i = p$ and $S_i = S$. 
The value of $a_1$, the maximum value with which $\mathcal{A}$ can be applied,
is equal to $\x{p}{S}/\beta_{p,S}$ for some pair $(p, S)$. We have
\[
a_1 = \frac{\x{p}{S}}{\beta_{p,S}^{\text{sum}}} \geq \frac{\x{p}{S}}{d \cdot \beta_{p,S}^{\max}} 
\geq \frac{a'}{d}.\qedhere
\]
\end{proof}

\subsection{Number of iterations of the main loop}

The procedure \func{FindAugmentation} takes a \bfa $\vec{x}$ and 
returns an augmentation that reduces the cost of the current solution by an amount which
is within $\Omega(1/d)$ of the best possible augmentation that can be applied to $\vec{x}$.
In order to bound  the number iterations of the main loop, we need to show that there always is
a good augmentation that can be applied to  $\vec{x}$ that moves it towards an optimal solution. 
The idea is that for any two assignments $\vec{x}$ and $\vec{y}$, $\vec{x}$ can   be
transformed into $\vec{y}$ by applying a sequence of augmentations.
Each augmentation decreases the number of variables in which $\vec{x}$ and $\vec{y}$
differ by one. Since the number of non-zero variables in any \bfa is at most $n+d$,
there are at most $2(n+d)$ augmentations in the sequence. Thus, if the difference in cost between $\vec{y}$ and $\vec{x}$
is $\Delta$, one of the augmentations will decrease the cost by at least $\Delta/2(n+d)$.
The idea is analogous to the partitioning the difference between two min cost flows into a set of
disjoint cycles. Some additional work is required to establish that the chosen augmentation can 
be applied with sufficient magnitude.

The proofs of Lemmas~\ref{lem:bigAug},~\ref{lem:bound-entries},~\ref{lem:bound-bfa} and \ref{lem:mainLoop} are given in the Appendix.

\begin{lemma}
\label{lem:bigAug}
Let $\vec{x}$ be a \bfa for an instance of the subset assignment problem
and let $\Delta$ be the difference in the objective function between $\vec{x}$ and the 
optimal solution. Then there is an augmentation $\mathcal{A}$ such that
when $\mathcal{A}$ is applied to $\vec{x}$ with the maximum possible
magnitude, the cost drops by at least
$\Delta/2(n+d)$.
\end{lemma}

In order to bound the number of iterations in the main loop, we need to know
the smallest difference in cost between two assignments that have different cost.

\begin{lemma}\label{lem:bound-entries}
If $A$ is an invertible $d \times d$ matrix with entries in $\{-1, 0, 1\}$ and $\vec{b}$ is a $d$-vector with integer entries, then 
there is an integer $\ell \leq d^{d/2}$
such that the solution $\vec{x}$ to $A\vec{x} = \vec{b}$ has entries of the form $k/\ell$ where $k$ is an integer. Moreover, if $\vec{b}$ also has entries in $\{-1, 0, 1\}$, then the entries of $x$ are at most equal to $d$.
\end{lemma}

The following bound comes from the fact that the fractionally assigned values are the solution
to a matrix equation with a $d \times d$ matrix over $\{-1, 0, 1\}$.

\begin{lemma}\label{lem:bound-bfa}
  If $\vec{x}$ is a \bfa, then there is an integer
$\ell \leq d^{d/2}$ such that every $\x{p}{S} = k/\ell$ for some integer $k$.
\end{lemma}

With this result, we can bound the number of iterations in our algorithm.

\begin{lemma}
\label{lem:mainLoop}
The number of iterations of the main loop is $O(n d^2 \log(dnC))$.
\end{lemma}

\subsection{Analysis of the running time}

The running time of $\func{Preprocess}$ is dominated by the running time of the main loop, so we
just analyze the running time of the main loop. 
To bound the size of the augmented binary search trees $\var{Moves}(\vec{v})$,
observe that for each $S$, there is at most one $T$ such that $\vec{T} - \vec{S} = \vec{v}$.
Therefore, the number of moves $(p, S, T)$ that can be stored in a single  tree is 
$O(2^d n)$. Updates are handled in logarithmic time, so the time per 
update to an entry in one of the trees is $O(d \log n)$.
Every time a variable $\x{p}{S}$ changes, there are $2^d$ subsets $T$ such that  the move
$(p, S, T)$ must be updated. In each iteration of the main loop there are $O(d)$ variable changes,
resulting in a total update time of $O(d^2 2^d \log n)$.

By Lemma~\ref{lem:bfaFracBound}, the  \bfa 
at the beginning of an iteration has at most $2d$ fractionally assigned variables.
An augmentation consists of at most $d+1$ moves and therefore changes the value of at most
$2(d+1)$ variables.
Thus, the input to \func{Restore} is  an assignment with $O(d)$ fractionally assigned variables.
Each iteration of \func{Restore} reduces the number of fractionally assigned variables by at least one.
Therefore, the number of iterations of \func{Restore} is  bounded by $O(d)$ and the total time spent 
in \func{Restore} during an iteration is $\poly(d)$.

The inner loop of $\func{FindAugmentation}$ requires $O(d)$ queries to one of the 
augmented binary search trees resulting in $O(d^2 \log n)$ time for each iteration of the inner loop.
The number of times the inner loop is executed is 
$\log(a_{\max}/a_{\min})$ times the number of augmentation profiles, $f(d)$. 
Therefore the running time of $\func{FindAugmentation}$ dominates
the running time of an iteration of the main loop which is
$O(f(d) d^2 \log n \log(a_{\max}/a_{\min}))$.
By Lemma~\ref{lem:mainLoop}, the number of iterations of the main loop is $O(nd^2 \log(dnC))$,
and since $f(d) \le \binom{3^d}{d+1}$,
the total running time is  
$O(\binom{3^d}{d+1} d^4 n \log n (\log n + \log C) \log(a_{\max}/a_{\min}))$.
We now 
bound $a_{\max}/a_{\min}$:

\begin{lemma}\label{lem:bound-a}
The values of $a$ are bounded above by $a_{\max} = d^{d/2} Z$
and below by $a_{\min} = 1/d^{d/2+1}$, where $Z = \max_p \size(p)$.
\end{lemma}

The proof of Lemma~\ref{lem:bound-a} is given in the Appendix. Hence, we get that $\log(a_{\max}/a_{\min})$ is $O(d^2 \log d \log Z)$ and the total running time is 
\[
O\left(\binom{3^d}{d+1} \poly(d) n \log(n) \log(nC) \log(Z)\right).
\]

\bibliography{refs}

\pagebreak

\appendix

\section{The algorithm \texorpdfstring{$\func{Restore}$}{Restore}}

The algorithm $\func{Restore}$ takes a feasible, perfectly filled assignment $\vec{x}$ 
and converts it to a \bfa that is also perfectly filled.
The cost of the resulting assignment is no larger than the cost of the input assignment.
The number of iterations is bounded by the number of fractionally assigned variables in
the input assignment.

\begin{algorithm}
\begin{pseudocode}
For each $p \in P_{\text{frac}}$, select an $S$ s.t. $\x{p}{S} > 0$. Call the chosen set $S_p$.\\
Let $X$ be the set\! of\! variables\! $\x{p}{S}$ s.t. $S \neq S_p$\! and\! $0 < \x{p}{S} < \size(p)$.\\
Order the variables in $X$: $\{\x{p_1}{S_1}, \ldots, \x{p_r}{S_r}\}$\\
Let $V$ be the set of vectors $\vec{S} - \vec{S}_p$ for each $\x{p}{S} \in X$.\\
\\
\WHILE $V$ is linearly dependent\\
\> Let $\beta_1, \ldots, \beta_r$ be such that $\sum_{i=1}^r \beta_i (\vec{S_i} - \vec{S}_{p_i}) = \vec{0}$.\\
\\
\> \IF $\sum_{i=1}^r \beta_i [ \cost(p_i,S_i) - \cost(p_i,S_{p_i})] > 0$\\
\>\> \FOR $i = 1, \ldots, r$\\
\>\>\> $\beta_i = -\beta_i$\\
\> $a = \min \left\{ \min_{i: \beta_i < 0} \left\{ \frac{\x{p_i}{S_i}}{- \beta_i} \right\},
\min_{i: \beta_i > 0} \left\{ \frac{\x{p_i}{S_{p_i}}}{ \beta_i} \right\} \right\}$\\
\> \FOR $i = 1, \ldots, r$\\
\>\> $\x{p_i}{S_{p_i}} = \x{p_i}{S_{p_i}} - a \cdot \beta_i$\\
\>\> $\x{p_i}{S_{i}} = \x{p_i}{S_{i}} + a \cdot \beta_i$\\
\\
\> \IF $\x{p}{S} \in X$ becomes $0$\\
\> \> remove $\x{p}{S}$ from $X$\\
\> \IF $\x{p}{S_p}$ becomes $0$\\
\>\>Select an $\x{p}{S'}$ from $X$ and remove it from $X$.\\
\>\>$S_p$ becomes $S'$.\\
\>\>Update vectors in $V$ with new $S_p$.
\end{pseudocode}
\caption{\func{Restore}}
\label{alg:restore}
\end{algorithm}

\pagebreak

\section{A preprocessing step for the generation of all augmentation profiles}

\begin{algorithm}
\begin{pseudocode}
$\mathcal{P} = \emptyset$\\
\FOR each subset $\{\vec{v}_1, \dots, \vec{v}_d\}$ of $\mathcal{V}=\{-1, 0, 1\}^d$\\
\> Let $V$ be an ordered list whose $i$-th element is $\vec{v}_i$.\\
\> Let $A$ be the matrix whose $i$-th column is $\vec{v}_i$.\\
\> Try to find $A$'s inverse.\\
\> \IF $A$ is not invertible\\
\> \> Continue.\\
\> \FOR each $\vec{w}$ in $\mathcal{V}$\\
\> \> $\vec{\alpha} = A^{-1}\vec{w}$\\
\> \> \FOR $i = 1, \ldots, d$\\
\> \> \> \IF $\alpha_i < 0$\\
\> \> \> \> Continue.\\
\> \> \> \IF $\alpha_i = 0$\\
\> \> \> \> Remove $\vec{v}_i$ from $V$.\\
\> \> \> \> Remove $\alpha_i$ from $\vec{\alpha}$.\\
\> \> Append $-\vec{w}$ to $V$.\\
\> \> Append $1$ to $\vec{\alpha}$.\\
\> \> Sort $V$ lexicographically.\\
\> \> Reorder $\vec{\alpha}$ to match $V$'s order.\\
\> \> Rescale $\vec{\alpha}$ so that the first component is 1.\\
\> \> Add $(V, \vec{\alpha})$ to $\mathcal{P}$.\\
\RETURN $\mathcal{P}$
\end{pseudocode}
\caption{\func{Preprocess}($d$)}
\label{alg:preprocess}
\end{algorithm}

\section{Proof of Lemma~\ref{lem:vecDiff}}

\begin{proof}
The first step is to come up with a linear combination of moves of the form
$(p, S, T)$ where $\x{p}{S} > \x[y]{p}{S}$ and $\x{p}{T} < \x[y]{p}{T}$
that transform $\vec{x}$ into $\vec{y}$.
The profile for the set of moves must be linearly dependent because the net change
to the load on each bin is $0$. However, the resulting profile is not necessarily minimally dependent.
The next step is to find a subset of those moves
that can be applied to $\vec{x}$ whose profile is
minimally dependent and whose $\vec{\alpha}$ has positive coefficients.
The following procedure accomplishes the first step:

\vspace{.1in}
\begin{pseudocode}
Initialize $\vec{z} = \vec{x}$ and $j = 1$.\\
\WHILE $\vec{z} \neq \vec{y}$\\
\> Find a $(p, S, T)$ such that $\x[z]{p}{S} > \x[y]{p}{S}$ and $\x[z]{p}{T} < \x[y]{p}{T}$.\\
\> $\beta_j = \min\{ \x[z]{p}{S} - \x[y]{p}{S},  \x[y]{p}{T} - \x[z]{p}{T}\}$\\
\> $(p_j, S_j, T_j) = (p, S, T)$\\
\> $\x[z]{p}{S} = \x[z]{p}{S} - \beta_j$\\
\> $\x[z]{p}{T} = \x[z]{p}{T} + \beta_j$\\
\> $j = j+1$
\end{pseudocode}
\vspace{.1in}

Define $\mathcal{S}$ to be the set of pairs $(p, S)$ such that $\x[z]{p}{S} > \x[y]{p}{S}$
and ${\cal T}$ to be the set of pairs $(p, T)$ such that $\x[z]{p}{T} < \x[y]{p}{T}$.
In each iteration, if $(p, S, T)$ is the selected move, then
either $(p, S)$ drops out of ${\cal S}$ or $(p, T)$ drops out of ${\cal T}$.
Therefore, the process is finite and a move is never selected twice.
Let $t$ be the number of moves selected in the process.
Applying each move $(p_j, S_j, T_j)$ with magnitude $\beta_j$  transforms
$\vec{x}$ into $\vec{y}$. Since $\vec{x}$ and $\vec{y}$ are both perfectly filled,
the net change to  the load on each bin is $0$:
\[\sum_{j=1}^t \beta_j (\vec{T}_j - \vec{S}_j) = \vec{0}.
\]

In the second step, we adjust the linear combination of moves selected until
its profile is minimally dependent. Let $B$ be the set of indices $j$ such that $\beta_j > 0$.
Initially $B = \{1, \ldots, t\}$.
Define $V_B = \{ \vec{T_j} - \vec{S_j}: j \in B\}$.
The following procedure accomplishes the second step:

\vspace{.1in}
\begin{pseudocode}
\WHILE there is a proper subset of $V_B$ that is linearly dependent\\
\> Select a  ${\bar B} \subseteq B$ such that $V_{\bar B}$ is minimally dependent.\\
\> Let $\{ \gamma_j: j \in {\bar B} \}$ be the unique set of values such that:\\
\>\> $\sum_{j \in {\bar B}} \gamma_j (\vec{T}_j - \vec{S}_j) = \vec{0}$,\\
\>\> and $\min_{j \in {\bar B}} \gamma_j = 1$.\\
\> \IF $\gamma_j > 0$, for every $j \in {\bar B}$\\
\>\> \RETURN $\{(p_j, S_j, T_j): j \in {\bar B}\}$\\
\> \ELSE\\
\>\> $c = \min_{j: \gamma_j < 0} \frac{\beta_j}{-\alpha_j}$\\
\> \FOR each $j \in {\bar B}$\\
\>\> $\beta_j = \beta_j + c \gamma_j$\\
\RETURN $\{(p_j, S_j, T_j): j \in B\}$\\
\end{pseudocode}
\vspace{.1in}

Note that since $\sum_{j \in {\bar B}} \gamma_j (\vec{T}_j - \vec{S}_j) = \vec{0}$,
adding a constant multiple of the sum $\sum_{j \in {\bar B}} \gamma_j (\vec{T}_j - \vec{S}_j)$
to $\sum_{j \in B} \beta_j (\vec{T}_j - \vec{S}_j) $, maintains the condition
that $\sum_{j \in B} \beta_j (\vec{T}_j - \vec{S}_j) = \vec{0} $.

The changes to $\vec{\beta}$ also maintain the condition that $\vec{\beta} \geq \vec{0}$.
If $\gamma_j > 0$, then adding $c \cdot \gamma_j$ to $\beta_j$ can only increase $\beta_j$.
For $j$ such that $\gamma_j < 0$, $c \leq -\beta_j/\gamma_j$, so
\[
\beta_j + c \gamma_j \geq \beta_j + \left( \frac{-\beta_j}{\gamma_j} \right) \gamma_j = 0.
\]
Since $c = -\beta_j/\gamma_j$ for some $j$, at least one $\beta_j$ becomes $0$ and the
set $B$ decrease by at least one index.
Therefore  $V_B$ eventually becomes a minimally dependent set, and the moves corresponding
to $j \in B$ satisfy the properties of being an augmentation.
\end{proof}

\section{Proof of Lemma~\ref{lem:bigAug}}

\begin{proof}
Let $\vec{y}$ be an optimal solution that is also a \bfa.
We define a sequence of assignments
$\vec{z}_0, \vec{z}_1, \ldots, \vec{z}_t$.
We start with $\vec{z}_0 = \vec{x}$ and
 describe how to obtain $\vec{z}_{j+1}$ from $\vec{z}_{j}$.
Let $\mathcal{S}_j$ be the set of pairs $(p, S)$ such that
$z_j(p, S) > y(p, S)$.
Let $\mathcal{T}_j$ be the set of pairs $(p, T)$ such that
$z_j(p, T) < y(p, T)$.
We know from Lemma~\ref{lem:vecDiff} that there is an augmentation
$\mathcal{A}_j$ that can be applied to $\vec{z}_j$ such that all the moves
are of the form $(p, S, T)$, where $(p, S) \in \mathcal{S}_j$
and $(p, T) \in \mathcal{T}_j$.
Apply augmentation $\mathcal{A}_j$ to $\vec{z}_j$ with magnitude $a_j$ to
get $\vec{z}_{j+1}$, where $a_j$ is the largest possible
magnitude with which $\mathcal{A}_j$
can be applied such that
$z_{j+1}(p, S) \geq y(p, S)$ for every $(p, S) \in \mathcal{S}_j$
and $z_{j+1}(p, T) \leq y(p, T)$ for every $(p, T) \in \mathcal{T}_j$.
Note that after the augmentation is applied,
it must be true that for some $(p, S) \in \mathcal{S}_j$, $z_{j+1}(p, S) = y(p, S)$
or for some $(p, T) \in \mathcal{T}_j$, $z_{j+1}(p, T) = y(p, T)$.
Therefore $\mathcal{S}_{j+1} \cup \mathcal{T}_{j+1}$ is a proper 
subset  of $\mathcal{S}_{j} \cup \mathcal{T}_{j}$.
Continue the process until $\mathcal{S}_{t} \cup \mathcal{T}_{t} = \emptyset$
which means that $\vec{z}_t = \vec{y}$.

Since any \bfa has at most $(n+d)$ non-zero variables,
$|\mathcal{S}_{0} \cup \mathcal{T}_{0}| \leq 2(n+d)$ and therefore $t \leq 2(n+d)$.
The $t$ augmentations cause the cost of the assignment to drop by $\Delta$,
so there is at least one augmentation that causes the cost to drop by at least
$\Delta/2(n+d)$. Suppose that the largest drop in cost  happens when 
$\mathcal{A}_j$ is applied with magnitude $a_j$ to $\vec{z}_j$.
We need to establish that $\mathcal{A}_j$ can be applied to $\vec{z}_0$ with magnitude
at least $a_j$. 

Consider a pair $({\bar p}, {\bar S})$ such that 
$z_j({\bar p}, {\bar S})$  decreases when $\mathcal{A}_j$ is applied to $\vec{z}_j$.
Then $({\bar p}, {\bar S}) \in \mathcal{S}_j$. Furthermore since each
$\mathcal{S}_{j} \subseteq \mathcal{S}_{j-1} \subseteq \cdots \subseteq \mathcal{S}_{0}$, then $({\bar p}, {\bar S}) \in \mathcal{S}_i$ for 
any $i$ in the range from $0$ to $j$.
The augmentations $\mathcal{A}_0, \ldots, \mathcal{A}_j$
can only
take mass off of $z_i({\bar p}, {\bar S})$ or leave it the same. 
Therefore $z_0({\bar p}, {\bar S}) \geq z_j({\bar p}, {\bar S})$ for any pair $({\bar p}, {\bar S})$
such that $\mathcal{A}_j$ causes $z_j({\bar p}, {\bar S})$ to decrease. 
Thus, if $\mathcal{A}_j$ can be applied
to $\vec{z}_j$ with magnitude $a_j$, then $\mathcal{A}_j$ can also be applied
to $\vec{z}_0$ with magnitude at least $a_j$.
\end{proof}

\section{Proof of Lemma~\ref{lem:bound-entries}}

\begin{proof}
As a consequence of the Laplace expansion of the determinant, the inverse of $A$ is equal to $\text{adj}(A)/\det(A)$, where $\text{adj}(A)$ is the adjugate matrix of $A$ or transpose of the cofactor matrix of $A$. Since $A$ has integer entries, so does its adjugate. This implies that the elements of $x$ are all of the form $k/\det(A)$ for some integer $k$. But by Hadamard's inequality, the determinant of $A$ is bounded by $d^{d/2}$. So this proves the first statement.
  
Now if the entries of $\vec{b}$ are in $\{-1, 0, 1\}$, then $k \leq d$. Since $A$ is invertible, its determinant is non-zero. Since its entries are integers, the absolute value of its determinant is at least 1. This proves the second statement.
\end{proof}

\section{Proof of Lemma~\ref{lem:bound-bfa}}

\begin{proof}
As in the definition of a \bfa, let $P_{\text{int}}$ denote the items that are integrally and $P_{\text{frac}}$ those that are partially assigned. Also, let $S_p$ be the subset or one of the subsets to which $p$ is assigned, let $F = \{(p,S): x_{p, S}>0 \text{ and } S \neq S_p\}$, and let $\vec{c}$ be the vector of bin capacities. Since $\vec{x}$ is perfectly filled, we have
\[
\sum_{p \in P_{\text{int}}} \x{p}{S_p} \vec{S}_p + \sum_{p \in P_{\text{frac}}} \x{p}{S_p} \vec{S}_p + \sum_{(p,S) \in F} \x{p}{S} \vec{S} = \vec{c}.
\]
Now $\x{p}{S_p} = \size(p) - \sum_{S\neq S_p} \x{p}{S}$ for all $p$, and in particular $\x{p}{S_p} = \size(p)$ for $p \in P_{\text{int}}$. So we have
\[
  \sum_{p} \size(p)\vec{S}_p + \sum_{(p, S) \in F} \x{p}{S} (\vec{S} - \vec{S}_p) = \vec{c}.
\]
By definition of a \bfa, the vectors $\vec{S} - \vec{S}_p$ are linearly independent. Therefore, if we extend this set to a basis of $\{-1, 0, 1\}$-vectors, we can view this equation as the matrix equation 
\[
  A\vec{x} = \vec{c} - \sum_{p} \size(p)\vec{S}_p
\]
where $A$'s columns are the basis vectors. Since $A$ has entries in $\{-1, 0, 1\}$ and the right hand vector has integer entries, an application of Lemma~\ref{lem:bound-entries} yields the result.
\end{proof}

\section{Proof of Lemma~\ref{lem:mainLoop}}

\begin{proof}
Let $C = \max_{p, S} \cost(p, S)$. The cost of the initial assignment is at most $nC$.
Since the costs are non-negative, the difference in cost between the initial assignment and
an optimal assignment is at most $nC$.

By Lemma~\ref{lem:bound-bfa}, for any \bfa $\vec{x}$, every $\x{p}{S}$ is an integer 
multiple of some $1/\ell$ where $\ell$ is an integer bounded by
$d^{d/2}$. Since the costs are integers, the cost of 
$\vec{x}$ is also an integer multiple of $1/\ell$.
Consider two \bfa's, $\vec{x}$ and $\vec{y}$  with different costs. The cost of $\vec{x}$ is a multiple of $1/\ell$
and the cost of $\vec{y}$ is  a multiple of $1/\ell'$, where $\ell$ and $\ell'$ are both
integers bounded by $d^{d/2}$. 
If $\ell = \ell'$,  then the difference in costs between $\vec{x}$ and $\vec{y}$
is at least $1/d^{d/2}$. If $\ell > \ell'$, the difference in cost is at least
\[
\frac{1}{\ell'} - \frac{1}{\ell} = \frac{\ell - \ell'}{\ell\ell'} \geq \frac{1}{d^d}.
\]

Lemma~\ref{lem:bigAug} indicates that if the difference in cost between the current assignment and the optimal assignment is $\Delta$, there is an augmentation that reduces
the cost by at least $\Delta/2(n+d)$ and Lemma~\ref{lem:findAug} indicates that the augmentation returned by \func{FindAugmentation}
reduces the cost by at least $1/2d$ times the best possible.
Therefore, each iteration reduces the difference in cost between the
current assignment and the optimal assignment by at least a factor of $1 - 1/(4d(n+d))$.
The number of iterations is the smallest $t$ such that
\[
nC \left( 1 - \frac{1}{4d(n+d)} \right)^t < \frac{1}{d^d},
\]
which is $O(nd^2 \log (d nC)) $.
\end{proof}

\section{Proof of Lemma~\ref{lem:bound-a}}

\begin{proof}
Recall that the entries of $\vec{\alpha}$ were obtained as the solution to the equation $A\vec{\alpha} = \vec{w}$ where the columns of $A$ and the vector $\vec{w}$ have entries in $\{-1, 0, 1\}$. By Lemma~\ref{lem:bound-entries}, $1/d^{d/2} \leq \alpha_i \leq d$.

An upper bound on $a$ is the ratio of the maximum possible value for $\x{p}{S}$ over the minimum possible value for $\alpha_i$. The highest value that $\x{p}{S}$ can achieve is $Z = \max_p \size(p)$ because of (\ref{eq:condition1}). So let $a_{\max} = d^{d/2} Z$.

Similarly a lower bound on $a$ is the ratio of the minimum possible $\x{p}{S}$ over the maximum possible value for $\alpha_i$. By Lemma~\ref{lem:bound-bfa}, $\x{p}{S}$ is least $1/d^{d/2}$. So we can take $a_{\min}$ to be $1/d^{d/2+1}$.
\end{proof}
\end{document}